\documentclass[10pt,conference]{IEEEtran}

\usepackage{graphicx} 

\usepackage{amssymb}
\usepackage{amsmath}

\usepackage{algorithmicx}
\usepackage{algpseudocode}
\usepackage{algorithm}

\newtheorem{definition}{Definition}
\newtheorem{proposition}{Proposition}
\newtheorem{theorem}{Theorem}
\newtheorem{lemma}{Lemma}
\newtheorem{example}{Example}
\newtheorem{assumption}{Assumption}

\newenvironment{proof}
{\noindent {\bf Proof} \\} {\hfill$\Box$}

\newcommand{\apart}{\#}
\newcommand{\wapart}{\#^w}

\newcommand{\separates}{\not \approx}
\newcommand{\notseparates}{\approx}

\newcommand{\dseparates}{\not \equiv}
\newcommand{\notdseparates}{\equiv}

\newcommand{\witness}[2]{w(#1,#2)}

\newcommand{\wset}[1]{W(#1)}

\newcommand{\diff}[1]{D(M,M_I,V,#1)}
\newcommand{\diffw}[1]{DW(M,M_I,V,W,#1)}

\newcommand{\absolute}[1]{|#1|}
\newcommand{\Real}{\mathbb{R}}

\newcommand{\faultdomain}{{\mathcal F}}

\newcommand{\faultdomainbounded}[1]{{\mathcal F}^{#1}}

\begin{document}

\title{Complete FSM Testing Using Strong Separability}



 \author{\IEEEauthorblockN{Robert M. Hierons}
 		\IEEEauthorblockA{Department of Computer Science\\
 			University of Sheffield, UK}
 		\and
 		\IEEEauthorblockN{Mohammad Reza Mousavi}
 		\IEEEauthorblockA{Department of Informatics\\
 		King's College London, UK}}


\maketitle

\begin{abstract}
Apartness is a concept developed in constructive mathematics, which has resurfaced as a powerful notion for separating states in the area of model learning and model-based testing. We identify some fundamental shortcomings of apartness in quantitative models, such as in hybrid and stochastic systems. We propose a closely-related alternative, called strong separability and show that using it to replace apartness addresses the identified shortcomings. We adapt a well-known complete model-based testing method, called the Harmonized State Identifiers (HSI) method, to adopt the proposed notion of strong separability. We prove that the adapted HSI method is complete.
As far as we are aware, this is the first work to show how complete test suites can be generated for quantitative models such as those found in the development of cyber-physical systems.
\end{abstract}

\section{Introduction}

Testing is an important part of the software development process but has traditionally been a largely manual process and so expensive and error-prone.
This has led to significant interest in techniques that support systematic, automated test generation.
Model-based testing (MBT) forms an  important class of such techniques in which test generation is based on a formal model (or \emph{specification}) that defines the set of  behaviours that a correct \emph{system under test (SUT)} may have.
There are many MBT techniques, with corresponding tools and evidence of effectiveness when applied to industrial systems (see,  e.g.,~\cite{chow78,fujiwara91,hennie64,luo94a,MillerS12,Tretmans08,vasilevskii73}).
We refer to survey papers for an overview of the field and its developments ~\cite{HieronsBBCDDGHKKLSVWZ09,lee94,Shafie2022,UttingPL12}.

In order to formally reason about the  effectiveness of an MBT technique, it is normal to assume a common semantic domain between the specification and the SUT; in other words, the SUT is assumed to behave like an unknown model $M_I$ that can be expressed using the same formalism as the specification $M$
\cite{gaudel95}.
Much of the MBT work has concerned state-based semantic domains: either labelled transition systems (LTSs) \cite{Tretmans08} or finite state machines (FSMs) \cite{lee96}.
Typically, the state-based models of the SUT are not produced by the developer/tester and are only hypothesised as the common semantic domain to ground the MBT approach on formal foundations.

This paper aims to widen the applicability of  FSM-based test techniques to a range of modern systems for which test generation is particularly challenging.
Interest in FSM-based testing
goes back to the seminal work of
Moore in 1956~\cite{moore56}, with Hennie introducing a test generation algorithm in 1964~\cite{hennie64}. 
The behaviour of an FSM is defined by transitions between states, with a transition having a corresponding input and output.
FSM-based test generation techniques typically aim to find two types of faults: output faults (a transition produces the wrong output); and state-transfer faults (a transition takes the SUT to the wrong state). State-transfer faults can lead to the SUT having more states than the specification FSM.
If there can only be output faults then it is sufficient to generate and execute tests that cover all transitions of the specification $M$.
However, this approach is known to be insufficient in situations in which there might be state-transfer faults~\cite{lee96}.

There are many FSM-based test generation techniques that aim to find state-transfer faults in addition to output faults \cite{BroyEtAl2005};
these techniques utilise input sequences that \emph{separate} states\footnote{An input sequence $\bar{x}$ separates two states $s_1$ and $s_2$ of $M$ if the input of $\bar{x}$ leads to different output sequences when applied in $s_1$ and $s_2$.} of the specification $M$ in order to \emph{identify} (or \emph{check}) states.
Most FSM-based test generations concern specification FSMs that are minimal, deterministic, and completely-specified\footnote{These terms are formally defined in Section~\ref{section:FSM}.} and we also consider such FSMs.
Test generation methods for other classes of FSMs are seen as future work and are discussed further in the Conclusions.

Of particular interest are test generation techniques that produce $m$-complete test suites:
test suites that are guaranteed to determine correctness as long as the SUT behaves like an (unknown) FSM that has at most $m$  states.
Possibly the best known such test generation technique is the W-method \cite{chow78,vasilevskii73},
which uses a characterising set to identify states (separate a state from all other states of the specification).
A key benefit of using a characterising set is that every (minimal, completely-specified, and deterministic) FSM has a characterising set and such a set can be found in low-order polynomial time.
Two other, more recent,
FSM-based test generation techniques,
the Wp-method~\cite{fujiwara91} and HSI-method~\cite{luo94a},
also use (prefixes of) sequences from a characterising set in order to identify states and return smaller $m$-complete test suites.

Now consider what happens if we wish to test a system in which an output is drawn from a continuous domain and observations are noisy. This is characteristic of cyber-physical systems \cite{KhakpourM15} and the conformance problem for such systems has been extensively researched \cite{BreugelHMW05,AbbasMF14,DeshmukhMP17}.   
This situation could occur, for example, in black-box testing of an autonomous system;
the outputs from a control system may well be discrete but these outputs are not actually observed in testing.
Instead, the tester uses sensors (e.g., radars and LIDARs) to observe the physical interaction of the autonomous system with its environment.
In such situations, observations are imprecise and we should not require that the specification and SUT produce the \emph{same} output when given an input sequence.
Testing has to take into account this imprecision and might do this by defining a metric $m$ between observations (trajectories) and using a threshold $t$ to define when two observations should be considered to be similar \cite{AlfaroFS09,BreugelHMW05}.
Then
two observations $o_1$ and $o_2$ are considered to be similar if $m(o_1,o_2) \leq t$ and an observation $o_2$ constitutes a failure if the expected observation $o_1$ is such that 
$m(o_1,o_2) > t$.
Separability is thus defined in terms of such a threshold being exceeded.

We obtain a similar situation when checking conformance between a stochastic state-based specification model and its implementation \cite{BreugelW01,BreugelHMW05,DESHARNAIS2002163}; this is a practical problem, e.g., in  the automotive domain and has been studied in the literature \cite{AraujoH0V20,QinEtAl2023}. In a stochastic FSM specification, for each state and input pair, a transition defines a distribution over outputs.
Testing with input sequence $\bar{x}$ then involves applying $\bar{x}$ multiple times,
obtaining a sample of values, and using a statistical test to check whether, with a stated confidence, that this sample is consistent with the specified distribution~\cite{HMM09}.
Typically, the statistical test will involve applying a measure (a test statistic) and comparing this against a fixed value (threshold) that depends upon the required confidence. 
So, again, we are comparing a measure against a threshold.
Note, however, that in this case there is the potential to reduce the threshold by increasing the number of times that a test (input) sequence is applied.

We need to reconsider what we mean by correctness (conformance) in scenarios such as those sketched above.
It no longer makes sense to require that, for each input sequence $\bar{x}$, the SUT and specification produce the \emph{same} output sequence.
Instead,
we might require that the output sequence produced by the SUT (i.e., observed in testing) is sufficiently \emph{similar} to that produced by the specification \cite{BreugelHMW05,AbbasMF14,DeshmukhMP17}.
This leads to a different conformance relation and so the proofs of completeness, for current FSM-based test techniques, do not apply.
In this paper we also show that the W, Wp, and HSI methods need not be complete when using such conformance relations.

The aim of the work described in this paper is to generalise FSM-based test techniques to conformance relations such as those described above.
One option might have been to adapt current FSM-based test techniques to a particular application domain.
However, we aimed to produce a more general solution that can be applied across a range of application domains.

We took as a starting point recent work by
Vaandrager \cite{Vaandrager24} that shows that it is possible to reason about the completeness of several FSM-based test techniques in terms of \emph{apartness}.
Apartness is a form of inequality
used in constructive mathematics.
Vaandrager \cite{Vaandrager24} showed how the notion of apartness can be used in the generation of a $m$-complete test set.
We found, however, that the notion of apartness is too strong for the scenario of interest in this paper (see Section~\ref{section:apartness}).

In this paper we weaken apartness in a way that allows us to reason about FSM-based testing when conformance is based on a metric $m$ and threshold $t$.
We start by specifying a notion of conformance defined in terms of a metric $m$ and threshold, which is reminiscent of the work on metric bisimulation \cite{BreugelW01,BreugelHMW05,DESHARNAIS2002163} and conformance testing \cite{AbbasMF14,DeshmukhMP17}.
We then define the notion of strongly separable, which is strictly weaker than apartness and stronger than separability, and discuss its relevance to the proposed conformance relation.
We consider the case where the states of the specification FSM $M$ are pairwise strongly separable.
We give a sufficient condition for a test suite to be $m$-complete, defined in terms of strong separability.
We adapt the  HSI-method to the scenarios we consider and show that the above sufficient condition implies that this version of the HSI-method returns $m$-complete test suites.
This
completeness result extends to the W-method and the Wp-method.
Note that these are exactly the FSM-based test generation techniques considered by Vaandrager~\cite{Vaandrager24}.

To illustrate the concepts throughout the paper and their potential applicability, we use the following hybrid systems model. The same example can be extended to cover stochastic aspects, but for the sake of readability we decided to use a minimal running example. 

\begin{figure}
    \centering
    \includegraphics[width=0.6\linewidth]{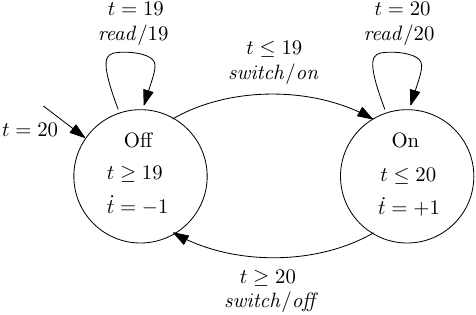}
    \caption{Continuous Thermostat System}
    \label{fig:thermostat-spec}
\end{figure}

\begin{example}\label{ex:thermostat}
Consider a thermostat system that senses the temperatures and uses it to switch a heater on/off and control the room temperature.
A diagrammatic representation of the thermostat behaviour is depicted in Figure \ref{fig:thermostat-spec}. The system starts off at time zero and the room temperature at 20 degree Celsius (in the discrete state Off) and  the temperature decreases linearly with a first derivative -1. When the temperature reaches 19 degrees Celsius, the display can be  updated through a ``read'' action (which outputs 19) and the heater is switched on through the action ``switch'',  which outputs ``on''.  The behaviour in the state labelled ``On'' is similar, with the main difference being that the temperature increases linearly with the first derivative equal to 1. 

Consider an implementation in which the initial condition and the guards differ by 0.5 degree.   We will compare such an implementation (as well as another  faulty one)  against  the specification  based on a notion of conformance  in the remainder of the paper.
The idea is to use an approximate notion of conformance that allows for an error margin and yet, produce test-cases that are complete, i.e., test cases that will fail on any non-conforming implementation with at most $m$ states. 
Consider, for example, a notion of conformance that has a threshold of 0.5 degrees Celsius for temperature differences (and no difference in the switching outputs) and thus, it accepts the implementation that differs only 0.5 degree Celsius in temperature with the specification in Figure \ref{fig:thermostat-spec}, while rejecting those with temperature differences that exceed the threshold or have incorrect switching behaviour.
In the rest of this paper we show how complete test suites can be produced for this specification.
\end{example}

The paper makes the following main contributions.

\begin{itemize}
\item We show that there are scenarios (classes of system, with a corresponding definition of conformance) in which the notion of apartness is too strong;
\item We define the notion of strong separability and demonstrate that this is applicable in the identified scenarios;
\item 
We give a sufficient condition for a test suite to be $m$-complete and show that test suites produced by the well-known Harmonized State Identifiers (HSI) method (as well as similar ones such as the Wp and W-methods) satisfy this condition as long as the specification is an FSM in which the states are pairwise strongly separable and, for state identification, we use sequences that strongly separate states; and
\item 
We illustrate all concepts through our running example and use this to show that the W, Wp and HSI methods are not complete if we do not adapt them (ie if we use separability rather than strong separability).
\end{itemize}

The paper is structured as follows.
We start, in Section~\ref{section:FSM}, by providing preliminary definitions.
In Section~\ref{section:apartness}
we extend classical FSM concepts and definitions to the scenarios of interest.
Within this we define separability and apartness and show that there are classes of system, including the Thermostat example, where separability is not an apartness relation.
This leads to us defining the notion of strong separability.
In Section~\ref{section:test} we give a sufficient condition, defined in terms of strong separability, for a test suite to be $m$-complete, with a consequence being that the well-known W, Wp, and HSI methods return $m$-complete test suites for the notion of conformance and class of system considered if we base the choice of state identifiers on strong separability rather than separability.
We also show that separability is insufficient for these test generation techniques: they can return test suites that are incomplete.
Finally, in Section~\ref{section:conclusions}, we draw conclusions and describe possible lines of future work.


\section{Preliminaries}\label{section:FSM}

In this section we define what we mean by a Finite State Machine and provide classical definitions.
In the next section we adapt this to the context of interest where, for example, outputs may be drawn from a metric space.

\begin{definition}\label{def:FSM}
A \emph{Finite State Machine (FSM)} $M$ is defined by a tuple $(S,s_0,X,Y,\delta,\lambda)$ in which:
    $S$ is the finite set of states;
    %
    $s_0 \in S$ is the initial state;
    %
    $X$ is the finite input alphabet;
    %
    $Y$ is the output alphabet;
    %
    $\delta: S \times X \rightarrow S$ is the transition function; and
    %
    $\lambda: S \times X \rightarrow Y$ is the output function.
\end{definition}

The transition function and the output function are lifted inductively to  sequences of inputs, by defining $\delta(s_i, \varepsilon) = s_i$, $\delta(s_i, x'.\bar{x}) = \delta(\delta(s_i, x'), \bar{x})$ and  $\lambda(s_i, \varepsilon) = \varepsilon$, $\lambda(s_i, x'.\bar{x}) = \lambda(s_i, x'). \lambda(\delta(s_i, x'), \bar{x})$, where $\varepsilon$ is the empty sequence, $x' \in X$, $\bar{x} \in X^*$. 
These FSMs are deterministic since for each state and input there is only one possible next state and output.
In addition,
the FSMs are completely-specified:
$\delta$ and $\lambda$ are total functions.
In this paper we restrict attention to deterministic and completely-specified FSMs, leaving other classes of FSM as a problem for future work.

In order to reason about testing and test effectiveness, it is normal to assume that the SUT behaves like an unknown model that can be expressed in the same formalism~\cite{gaudel95}.

\begin{assumption}
The behaviour of the SUT can be represented by an FSM $M_I = (Q,q_0,X,Y,\delta_I,\lambda_I)$ with the same input and output alphabets as the specification FSM $M$.
\end{assumption}

We don't assume that the FSM $M_I$ is known, only that there is some such $M_I$.

Note that the Definition~\ref{def:FSM} relaxes the standard definition of an FSM, which requires that the output alphabet $Y$ is finite.
This may not seem to be an important point since the set of outputs that can be produced by $M$ is finite since $M$ is deterministic and has finitely many states and inputs.
However, allowing $Y$ to be infinite has an impact on the set of potential SUTs being tested, which becomes infinite (even if we bound the number of states), and so the set of possible faults.

\begin{figure}
    \centering
    \includegraphics[width=\linewidth]{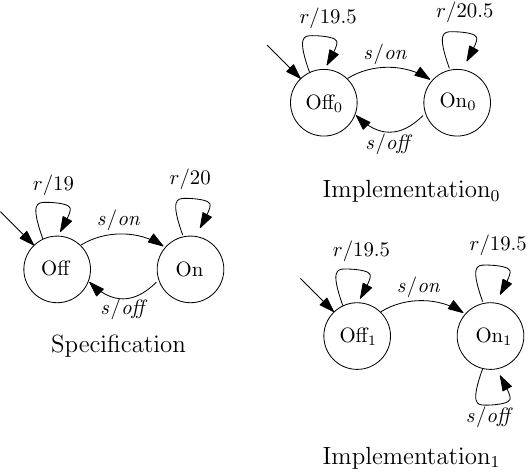}
    \caption{Discretised Thermostat System: A Specification and Two Implementations}
    \label{fig:thermostat-spec-discretised}
\end{figure}

\begin{example}\label{ex:thermostat-discretised}
Consider our thermostat in Example \ref{ex:thermostat} (Figure \ref{fig:thermostat-spec}); assume that we would like to focus on the input/output behaviour and abstract from the internal continuous dynamics of the states. The FSM arising from the specification in Figure \ref{fig:thermostat-spec} is depicted in Figure \ref{fig:thermostat-spec-discretised}. 
The FSM (as well as the two implementations depicted in the same figure), have  the  set of input symbols $\{r, s\}$, representing reading the temperature and switching the heater, respectively. The output set is $\Real \cup \{on, off\}$, i.e.,  the (uncountable) set of real numbers for temperature and the status values $\mathit{on}$ and $\mathit{off}$. 
In our example, the relaxation to allow for infinite outputs is essential here to represent the infinite set of possible faulty implementations with deviating outputs, while a finite set of inputs would be sufficient to represent any discretisation of the passing of time. 
\end{example}

We use $\faultdomain$ to denote the set of FSMs with input alphabet $X$ and output alphabet $Y$.
Further, given integer $m$ we use $\faultdomainbounded{m}$ to denote the set of FSMs with input alphabet $X$, output alphabet $Y$, and at most $m$ states.
As noted, since $Y$ need not be finite, the set $\faultdomainbounded{m}$ might be infinite.
The following shows that if $Y$ is finite then so is $\faultdomainbounded{m}$;
from $\faultdomainbounded{m}$ being finite it is not difficult to show that there must be a finite $m$-complete test suite.

\begin{proposition}
Let us suppose that $M = (Q,q_0,X,Y,\delta,\lambda)$ is an FSM and $Y$ is finite.
Given integer $m$, we have that $\faultdomainbounded{m}$ is finite.
\end{proposition}

\begin{proof}
First, observe that each FSM in $\faultdomainbounded{m}$ can be seen as being defined by a directed graph that shows the transitions and a mapping from each edge of the directed graph to a corresponding input $x \in X$ and output $y \in Y$ pair.
In the directed graph, each vertex has outdegree $|I|$ and there are only finitely many such directed graphs.
In addition, for each such directed graph there are finitely many mappings from edges to input/output pairs.
The result therefore holds.
\end{proof}

There is a correspondence between FSMs and states of FSMs:
if $M$ is an FSM and $s_i$ is a state of $M$ then we can also see $s_i$ as being an FSM:
the FSM formed by making $s_i$ the initial state of $M$.
We can thus use $\faultdomain$ to denote the set of possible states of the unknown FSM $M_I$ that represents the SUT.


In traditional FSM-based testing, 
testing typically involves applying an input sequence to the SUT and checking that the output sequence produced is that specified.
We can formalise this in terms of separating states and FSMs; we use the term d-separate in this paper in order to distinguish between the classical notion and that required for the scenarios of interest (defined in the next section).

\begin{definition}
Given FSM $M$ and states $s_1, s_2 \in S$,
an input sequence $\bar{x}$ \emph{d-separates} $s_1$ and $s_2$ if and only if
$\lambda(s_1,\bar{x}) \neq \lambda(s_2,\bar{x})$.
If $\bar{x}$ \emph{d-separates} $s_1$ and $s_2$
then we can write  $s_1 \dseparates_{\bar{x}} s_2$
and say that $s_1$ and $s_2$ are d-separable.
If $\bar{x}$ does not d-separate $s_1$ and $s_2$ then we write
$s_1 \notdseparates_{\bar{x}} s_2$.
Similarly, given a set $W$ of input sequence,
we write 
$s_1 \dseparates_{W} s_2$
if some input sequence in $W$ d-separates $s_1$ and $s_2$ and otherwise we write
$s_1 \notdseparates_{W} s_2$.
\end{definition}

It is straightforward to see that $\notdseparates_{\bar{x}}$ is an equivalence relation.
In addition, it can be applied when comparing two FSMs $M$ and $M_I$ since one can define an FSM that is the disjoint union of $M$ and $M_I$.

We get a corresponding notion of conformance, used in traditional FSM-based testing,
which we call d-conformance.

\begin{definition}
Given FSMs $M= (S,s_0,X,Y,\delta,\lambda)$ and $M_I= (Q,q_0,X,Y,\delta_I,\lambda_I)$,
we have that $M_I$ \emph{d-conforms} to $M$
if and only if for all $\bar{x} \in X^*$
we have that
$M_I \notdseparates_{\bar{x}} M$.
\end{definition}


\section{Conformance, Apartness and Strong Separation}\label{section:apartness}

In this section we start by explaining how the notions of separability and conformance can be naturally translated to a range of scenarios in which outputs could be drawn from
a metric space.
We then define strong separability.

\subsection{Separability and Conformance}

In the scenarios of interest 
an output might, for example, be the speed of a vehicle as measured by a sensor, which is an estimate of the actual speed.
An output in $Y$ could also be a probability distribution;
in testing we then apply a test case multiple times, observe a set of values drawn from an unknown distribution, and use statistical hypothesis testing techniques.
In such situations, we need some way of comparing the observations made with the expected output and we formalise this as a similarity relation $\sim$ on outputs.
As previously explained, such a notion of similarity will typically be  defined in terms of a metric $m$ and a threshold $t$  (below, we use the running example to illustrate this).
Note that $\sim$ normally will not be transitive and therefore need not be an equivalence relation.

The following lifts $\sim$ to  output sequences in the natural way (we use $\varepsilon$ to denote the empty sequence):

\begin{definition}
Given similarity relation $\sim$ on $Y$, $y, y' \in Y$, and $\bar{y}, \bar{y}' \in Y^*$
we have that:
\begin{itemize}
    \item $\varepsilon \sim \varepsilon$;
    \item $y \bar{y} \sim y' \bar{y}'$ if and only if $y \sim y'$ and 
    $\bar{y} \sim \bar{y}'$;
\end{itemize}
\end{definition}

We can define conformance in terms of $\sim$ as follows. 

\begin{definition}
Given FSMs $M= (S,s_0,X,Y,\delta,\lambda)$ and $M_I= (Q,q_0,X,Y,\delta_I,\lambda_I)$,
we define that $M_I$ \emph{conforms} to $M$
if and only if for all $\bar{x} \in X^*$
we have that
$\lambda(s_0,\bar{x}) \sim \lambda_I(q_0,\bar{x})$.
\end{definition}

\begin{example}\label{ex:similarity}
Consider the FSMs given in Figure \ref{fig:thermostat-spec-discretised}; we
define the distance between real-valued outputs to be the absolute difference between them; the distance between $\mathit{on}$ and $\mathit{off}$ (as well as the distance between real numbers on one hand, and status values) is defined to be $\infty$. 
Define two outputs $o, o' \in \Real$ to be $\epsilon$-conforming, denoted by $t \sim_\epsilon t'$, when $\absolute{o - o'} \leq \epsilon$. This leads to a natural notion of conformance, inspired by the hybrid conformance literature \cite{AbbasMF14}, on FSMs. 

According to this notion of conformance, $\mathit{Specification}$ $\sim_{0.5}$ $\mathit{Implementation}_0$.
This conformance  relation holds, because for all  input sequences,  the difference between the respective outputs is bounded by  $0.5$.

Consider the specification
$\mathit{Specification}$ and implementation $\mathit{Implementation}_1$ in
Figure \ref{fig:thermostat-spec-discretised}.
It does \emph{not} hold that 
$\mathit{Specification}$ $\sim_{0.5}$ $\mathit{Implementation}_1$. In fact, 
$\mathit{Implementation}_1$
does not $\epsilon$-conform to $\mathit{Specification}$ for any value of $\epsilon$, witnessed by the output to the input sequence $\mathit{s}. \mathit{s}. \mathit{s}$, in which the difference between the last output in the two systems is $\infty$. 
\end{example}

We can also define the notion of separating two states for the scenario of interest in which $Y$ need not be finite and there is a similarity relation on outputs.

\begin{definition}
Given FSM $M$ and states $s_1, s_2 \in S$,
an input sequence $\bar{x}$ \emph{separates} $s_1$ and $s_2$ if and only if
$\lambda(s_1,\bar{x}) \not \sim \lambda(s_2,\bar{x})$.
If $\bar{x}$ \emph{separates} $s_1$ and $s_2$
then we can write  $s_1 \separates_{\bar{x}} s_2$.
If $\bar{x}$ does not separate $s_1$ and $s_2$ then we write
$s_1 \notseparates_{\bar{x}} s_2$.
Similarly, given a set $W$ of input sequences,
we write 
$s_1 \separates_{W} s_2$
if there is a input sequence in $W$ that separates $s_1$ and $s_2$ and otherwise we write
$s_1 \notseparates_{W} s_2$.
\end{definition}

Where the input sequence of interest is clear,
we often use $\separates$ and $\notseparates$, avoiding reference to the actual input sequences used to separate states.

Conformance and separability are related.

\begin{proposition}\label{prop:conformance-separability}
Given FSMs $M= (S,s_0,X,Y,\delta,\lambda)$ and $M_I= (Q,q_0,X,Y,\delta_I,\lambda_I)$,
$M_I$ conforms to $M$
if and only if
no input sequence separates $s_0$ and $q_0$.
\end{proposition}

When two FSMs, $M$ and $M_I$ do not conform to each other, we denote this by  $M \separates M_I$; following Proposition \ref{prop:conformance-separability}, this can only happen if there exists $\bar{x} \in X^*$ such that the initial states of $M$ and $M_I$ are separated, which we denote by $M \separates_{\bar{x}} M_I$.

\begin{example}\label{ex:conformance}
Consider again the FSMs given in Figure \ref{fig:thermostat-spec-discretised} and the notion of similarity $\sim_{0.5}$; we have that the states $\mathit{Off}$ in  $\mathit{Specification}$ and $\mathit{Off}_0$ in  $\mathit{Implementation}_0$ are not separable by any input sequence and the two FSMs conform to each other with respect to $\sim_{0.5}$. 

However, $\mathit{Specification}$ and  $\mathit{Implementation}_1$ do not conform to each other. The input sequence $s.s.s$ separates the states $\mathit{Off}$ and $\mathit{Off}_1$, because  the former  produces the output sequence $\mathit{on}$.$\mathit{off}$.$\mathit{on}$, while  the later  produces  $\mathit{on}$.$\mathit{off}$.$\mathit{off}$ and the two output sequences have a distance of $\infty$. 
\end{example}


In order to reason about test effectiveness, we will use the notion of the \emph{product machine} defined by a specification FSM $M$ and an unknown FSM $M_I$ that models the SUT.

\begin{definition}
Given FSMs $M = (S,s_0,X,Y,\delta,\lambda)$ and $M_I = (Q,q_0,X,Y,\delta_I,\lambda_I)$,
the product machine $P(M,M_I)$ is the FSM $(S \times Q \cup \{e\}, (s_0,q_0),X,Y \cup \{err\},\delta_P,\lambda_P)$ for $e \notin S \times Q$ and $err \notin Y$
such that for every state $(s_i,q_i) \in S \times Q$ of $P(M,M_I)$ and input $x \in X$ we have that
\begin{itemize}
    \item If $\lambda(s_i,x) \sim \lambda_I(q_i,x)$ then
    $\delta_P((s_i,q_i),x) = (\delta(s_i,x),\delta_I(q_i,x))$ and $\lambda_P((s_i,q_i),x)) = \lambda_I(q_i,x)$; 
    \item else $\delta_P((s_i,q_i),x) = e$
    and $\lambda_P((s_i,q_i),x)) = err$.
\end{itemize}
Further, for all $x \in X$ we have that
$\delta_P(e,x) = e$
and $\lambda_P(e,x) = err$.   
\end{definition}

The following shows how the product machine relates to finding faults in testing.

\begin{proposition}\label{prop_error_state}
Given FSMs $M$ and $M_I$,
an input sequence $\bar{x}$ separates $M$ and $M_I$ if and only if $\bar{x}$ takes the product machine $P(M,M_I)$ to the error state.
\end{proposition}

\begin{proof}
We will use proof by induction on the length of the input sequences considered.
The result holds immediately for the base case, which is the empty sequence,
since the empty sequence cannot separate two FSMs and also cannot take a product machine to the error state.

Inductive hypothesis:
the result holds for all input sequences of length less than $k$ ($k > 0$).
We therefore assume that $\bar{x}$ has length $k$ and are required to prove that
$M \separates_{\bar{x}} M_I$
if and only if $\bar{x}$ takes $P(M,M_I)$
to the error state.
We consider the two directions:

\begin{enumerate}
\item Assume that $M \separates_{\bar{x}} M_I$
and we are required to prove that $\bar{x}$ takes $P(M,M_I)$
to the error state.

If there is a proper prefix $\bar{x}_1$ of $\bar{x}$ such that $M \separates_{\bar{x}_1} M_I$ then we can apply the inductive hypothesis to $\bar{x}_1$ and we obtain that
$\bar{x}_1$ takes $P(M,M_I)$
to the error state.
Thus, $\bar{x}$ takes $P(M,M_I)$
to the error state as required.
We will therefore assume that no proper prefix of $\bar{x}$ separates $M$ and $M_I$.

Since no proper prefix of $\bar{x}$ separates $M$ and $M_I$,
there is an input $x$ and input sequence $\bar{x}_1$ such that
$\bar{x} = \bar{x}_1 x$ and $M \notseparates_{\bar{x}_1} M_I$.
By the inductive hypothesis,
$\bar{x}_1$ does not take the product machine to the error state.

Let $(s_i,q_i)$ denote the state of the Product Machine reached by $\bar{x}_1$.
Since
$M \separates_{\bar{x}} M_I$
and
$M \notseparates_{\bar{x}_1} M_I$
we must have that $x$ separates states $s_i$ and $q_i$.
But, by the definition of the Product Machine,
this means that the transition with starting state $(s_i,q_i)$ and input $x$ takes the Product Machine to the error state as required.

\item Assume that $M \notseparates_{\bar{x}} M_I$
and we are required to prove that  $\bar{x}$ does not take $P(M,M_I)$ to the error state.
Since $\bar{x}$ is non-empty, $\bar{x} = \bar{x}_1 x$ for some input $x$ and input sequence $\bar{x}_1$.
Further, $M \notseparates_{\bar{x}_1} M_I$
and so, by the inductive hypothesis,
$\bar{x}_1$ does not take $P(M,M_I)$ to the error state.
Let us suppose that
$\bar{x}_1$ takes $P(M,M_I)$ to state $(s_i,q_i)$.
But,
since $M \notseparates_{\bar{x}} M_I$,
by definition we have that
$s_i \notseparates_{x} q_i$.
Thus, by the definition of the Product Machine,
there is a transition from $(s_i,q_i)$ to
$(\delta(s_i,a),\delta_I(q_i,x))$ with input $x$.
Thus,
$\bar{x}$ does not take $P(M,M_I)$ to the error state as required.
\end{enumerate}
\end{proof}

The following is immediate from the above.

\begin{proposition}
Given FSMs $M$ and $M_I$,
$M_I$ conforms to $M$ if and only if
the error state of $P(M,M_I)$
is not reachable.
\end{proposition}

\begin{figure}
    \centering
    \includegraphics[width=\linewidth]{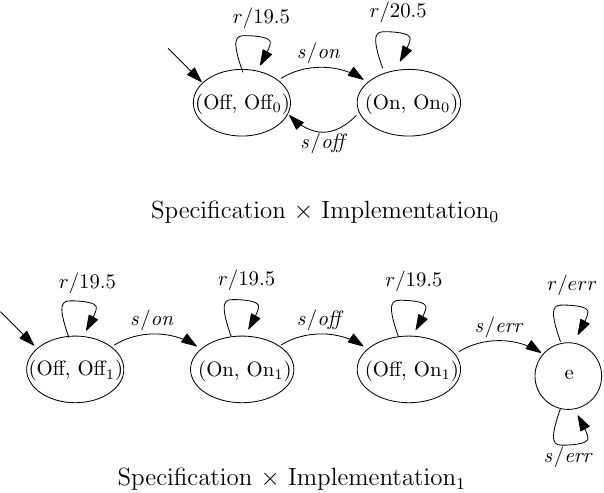}
    \caption{Thermostat System: Product Machines of Conforming and Non-Conforming Implementations}
    \label{fig:thermostat-product}
\end{figure}

\begin{example}
Consider the two implementations analysed for conformance in Example \ref{ex:similarity}. Figure \ref{ex:similarity} depicts the product automaton of the two implementation. As it 
can be observed in the figure, the product machine concerning $\mathit{Implementation}_0$, does not feature any reachable error state, because all transitions can be matched up to the error margin (of 0.5). However, for  $\mathit{Implementation}_1$, due to the dissimilar behaviour observed after the third switch transition, the product machine will move to an error state.     
\end{example}

\subsection{Apartness and Strong Separability}

Previous work has used the notion of an apartness relation and shown how one can use this to reason about test effectiveness when there is a finite set of outputs and conformance is defined in terms of equality of outputs rather than similarity.

\begin{definition}
A binary relation $\apart$ on a set $Z$ is an \emph{apartness relation}
if and only if it satisfies the following properties
\begin{itemize}
    \item It is irreflexive: for all $x \in Z$ we have that $x \apart x$ does not hold;
    \item It is symmetric: for all $x,y \in Z$, if $x \apart y$ then $y \apart x$; and
    \item It is co-transitive: for all $x, y, z \in Z$, if $x \apart y$ then either $x \apart z$ or $y \apart z$.
\end{itemize}
\end{definition}

Consider the classical notion of separability for FSMs defined in Section~\ref{section:FSM} (d-separability).
It is immediately clear that d-separability is both irreflexive and symmetric.
Now let us suppose that input sequence $\bar{x}$ d-separates states $s_1$ and $s_2$ of $M$ and so $\bar{x}$ leads to different output sequences $\bar{y}_1$ and $\bar{y}_2$ when applied in $s_1$ and $s_2$.
Further, let us suppose that we input $\bar{x}$ in some other state $s_i$
and this leads to output sequence $\bar{y}$.
Since
$\bar{x}$ d-separates $s_1$ and $s_2$ we have that
$\bar{y}_1 \neq \bar{y}_2$.
We therefore must have that either 
$\bar{y} \neq \bar{y}_1$ or $\bar{y} \neq \bar{y}_2$ (or both)
and so $\bar{x}$ must d-separate $s_i$ from at least one of $s_1$ and $s_2$.
As a result, d-separability is  co-transitive.
From this it is clear that d-separability is an apartness relation when
testing a deterministic SUT against a completely-specified deterministic FSM $M$,
with identify rather than similarity,
and so the results of Vaandrager \cite{Vaandrager24} can be applied.

Now consider what we mean by separability when we have continuous outputs and similarity is defined in terms of a metric $m$ and a threshold $t$.
Then an input sequence separates states $s_1$ and $s_2$ if the resultant observations $o_1$ and $o_2$ are such that
$m(o_1,o_2) > t$.
In such scenarios, separability is not an apartness relation since it need not be co-transitive, as shown by the following. 

\begin{example}\label{ex:seaparable1}
Consider the FSM thermostat models in Example \ref{ex:thermostat-discretised} (Figure \ref{fig:thermostat-spec-discretised}); there, consider the states $\mathit{Off}$ and $\mathit{On}$ in the  
$\mathit{Specification}$ FSM. 
They are separable using both inputs $r$ and $s$.
Let us consider what happens if we use $r$ to separate these states and we have some state $s_i$ of the SUT.
Then $s_i$ can have $r/19.5$, e.g., the state $\mathit{Off}_0$ in $\mathit{Implementation}_0$, that cannot be separated using only input $r$ from either of the two states.
This shows that separability is not co-transitive and so is not an apartness relation.
\end{example}

Since we cannot use the notion of apartness, we explore how this can be weakened for the scenarios of interest.
We first discuss how d-separability is used in classical FSM-based test techniques.

In most FSM-based test techniques that aim to find state-transfer faults, input sequences that d-separate states of a specification FSM are used to check the current state of the SUT.
To see how this works, let us suppose that  input sequence $\bar{x}$ d-separates states $s_1$ and $s_2$ of the specification FSM
and that the (expected) output sequences produced are $\bar{y}_1$ and $\bar{y}_2$ respectively.
Let us suppose that $\bar{x}$ is applied to the SUT when the SUT is in some an state $s_i \not \in \{s_1,s_2\}$.  
There are three possibilities:

\begin{enumerate}
\item The SUT produces output sequence $\bar{y}_2$ in response to $\bar{x}$.
We thus have that $\bar{x}$ d-separates $s_i$ and $s_1$.
\item The SUT produces output sequence $\bar{y}_1$ in response to $\bar{x}$.
We thus have that $\bar{x}$ d-separates $s_i$ and $s_2$.
\item The SUT produces some other output sequence $\bar{y}_3$ in response to $\bar{x}$.
We thus have that $\bar{x}$ d-separates $s_i$ and from both $s_1$ and $s_2$.
\end{enumerate}

FSM-based test techniques utilise such information that is provided by applying input sequences that d-separate the states of the specification.
However,
in Example~\ref{ex:seaparable1} we saw that
the above does not hold in the scenarios that we consider in this paper:
if two states of the specification are separated by an input sequence $\bar{x}$ then it is possible that a state of the SUT is not separated from either of them by $\bar{x}$.

We now strengthen the notion of separability, to what we call strong separability,
in a manner that ensures that strong separability satisfies the property (of d-separability in traditional FSM-based testing scenarios) described above.

\begin{definition}\label{def:strong_sep}
Given FSM $M$ and states $s_1, s_2 \in S$,
input sequence $\bar{x}$ is a \emph{witness} that 
states $s_1$ and $s_2$, $s_1 \neq s_2$, are \emph{strongly separable},
denoted $s_1 \wapart_{\bar{x}} s_2$,
if and only if for all $s_3 \in \faultdomain$
we have that either
$s_1 \separates_{\bar{x}} s_3$
or
$s_2 \separates_{\bar{x}} s_3$.
We also say that $s_1$ and $s_2$ are \emph{strongly separable} and denote this $s_1 \wapart s_2$.
\end{definition}

\begin{example}\label{ex:seaparable}
Consider the FSM thermostat models in Example \ref{ex:thermostat-discretised} (Figure \ref{fig:thermostat-spec-discretised}).
Recall that  states $\mathit{Off}$ and $\mathit{On}$ in the  
$\mathit{Specification}$ FSM are separable using both inputs $r$ and $s$.
However, we have seen that they are not strongly separable using input $r$. Namely, there can be a state that has $r/19.5$, e.g., the state $\mathit{Off}_0$ in $\mathit{Implementation}_0$, that cannot be separated using only input $r$ from either of the two states. In such cases, one can fix this issue, i.e., obtain strong separability, by either choosing different input sequences (in this case $s$) or decreasing the threshold of similarity, e.g., to $0.1$, so that all possible states in $\faultdomain$ are separable (by $r$) from at least one of these two states. 
In practice, the use of a lower threshold might require changes to the equipment used to observe the SUT in testing or to additional test runs with the same input sequence (so that more precise estimates are produced).
\end{example}

The following proposition states and proves the basic properties of strong separability. 

\begin{proposition}
Given FSM $M$ and input sequence $\bar{x}$, the relation $\wapart_{\bar{x}}$ is symmetric and irreflexive.
\end{proposition}

\begin{proof}
The symmerry of relation $\wapart_{\bar{x}}$ is immediate from Definition~\ref{def:strong_sep}.
To prove that $\wapart_{\bar{x}}$ is irreflexive, we use a proof by contradiction and assume that there is an input sequence $\bar{x}$ and state $s_i$ such that $s_i \wapart_{\bar{x}} s_i$.
Choosing $s_3 = s_i$ we then obtain,
from Definition~\ref{def:strong_sep},
that $s_i \separates_{\bar{x}} s_i$.
However, this contradicts $\separates_{\bar{x}}$ being irreflexive and so the result follows.
\end{proof}

Using the following proposition, we can compare strongly separability and separability.

\begin{proposition}
Given FSM $M$, input sequence $\bar{x}$ and states $s_1$ and $s_2$ of $M$ such that $s_1 \wapart_{\bar{x}} s_2$, we have that $\bar{x}$ separates $s_1$ and $s_2$.
\end{proposition}

\begin{proof}
Choose $s_3 = s_1$.
By the definition of $\wapart_{\bar{x}}$
(Definition~\ref{def:strong_sep}),
we must have that either $s_1 \separates_{\bar{x}} s_1$ or $s_2 \separates_{\bar{x}} s_1$.
Since $\separates_{\bar{x}}$ is irreflexive, we therefore have that $s_2 \separates_{\bar{x}} s_1$ and so the result follows.   
\end{proof}

The following is immediate from the definition of strong separability and corresponds to the property (described above) of d-separability in traditional FSM-based testing.

\begin{proposition}
Given states $s_1$ and $s_2$ of $M$ with $s_1 \wapart_{\bar{x}} s_2$, if $s_3 \in \faultdomain$ then at least one of the following must hold.
\begin{enumerate}
\item $\bar{x}$ separates $s_3$ from $s_1$;
\item $\bar{x}$ separates $s_3$ from $s_2$;
\item $\bar{x}$ separates $s_3$ from both $s_1$ and $s_2$.
\end{enumerate} 
\end{proposition}


\section{Test Generation}\label{section:test}

In this section, we explore the problem of test generation when,
for every pair of states of specification $M$, we have a specific input sequence that provides a witness that these states are strongly separable.

\begin{assumption}
For every pair $s_1, s_2$ of distinct states of $M$, there is an input sequence $\bar{x}$ that is a witness that $s_1$ and $s_2$ are strongly separable. 
We will use $\witness{s_1}{s_2}$ to represent such a witness
(and so $s_1 \wapart_{\witness{s_1}{s_2}}  s_2$)
and require that
$\witness{s_1}{s_2} = \witness{s_2}{s_1}$.
\end{assumption}

We will retain the notion of a state cover, used in classical FSM-based test generation techniques.
First we define a state cover, which is a 
minimal prefix-closed set of input sequences that, between them, reach all states of $M$.

\begin{definition}
Given a finite state machine $M$ whose states are pairwise strongly separable, a set $V$ of input sequences is  a \emph{state cover} for $M$
if $V$ is prefix closed and for all $s_i \in S$ there is exactly one input sequence $v_i \in V$ such that
$\delta(s_0,v_i) = s_i$.
\end{definition}

\begin{example}\label{ex:state-cover}
Consider the $\mathit{Specification}$ FSM in Example \ref{ex:thermostat-discretised} (Figure \ref{fig:thermostat-spec-discretised}); a state cover set $V$ for $\mathit{Specification}$ is $\{\varepsilon, s\}$. 
\end{example}

If we have a set of test sequences $Z$, we can remove from $Z$ any sequence that is a proper prefix of another sequence in $Z$.
We will let $RP(Z)$ denote the set of input sequences produced by removing from $Z$ any input sequence that is a proper prefix of another input sequence in $Z$.

\begin{definition}
Given a set $Z$ of input sequences,
its irredundant subset,
denoted by $RP(Z)$, is defined as follows:
\[
RP(Z) = \{ \sigma_1 \in Z |  \nexists \sigma_2 \in X^* \setminus \{\varepsilon\} .
\sigma_1.\sigma_2 \in Z
\}
\]
\end{definition}

In the following,
for a state $s_i$ of $M$ we use
$\wset{s_i}$ to denote the set of
input sequences used to show that $s_i$ is strongly separable from other states of $M$.
Recall that for a pair $s_1, s_2$
of distinct states of $M$,
$\witness{s_1}{s_2}$ denotes
a unique input sequence that is a witness for $s_1$ and $s_2$ being apart.

\begin{definition}
Given state $s_1$ of $M$,
we let $\wset{s_1}$ denote the \emph{state identification} set
\[
\wset{s_1} = RP(\{\witness{s_1}{s_2} | s_2 \in S \setminus \{s_1\} \})
\]
\end{definition}

\begin{example}\label{ex:state-identification}
Consider the $\mathit{Specification}$ FSM in Example \ref{ex:thermostat-discretised} (Figure \ref{fig:thermostat-spec-discretised}).
Both $r$ and $s$ pairwise separate the states of $\mathit{Specification}$ but we have seen that only $s$ strongly separates these states:
$\mathit{Off} \wapart_{s} \mathit{On}$
and
$\neg \mathit{Off} \wapart_{r} \mathit{On}$.
Since $\mathit{Off} \wapart_{s} \mathit{On}$, we can choose the state identification set for $\mathit{Off}$ (and symmetrically for $\mathit{On}$) to be $\{s\}$. 
\end{example}

Recall that we do not know the FSM $M_I$ that models the SUT but we will want to reason about the states of this unknown FSM that are met during testing.  
Given a state $s_i$ of $M$, we use $B(s_i)$ to denote the ball (set) of states of $M_I$ that are not separated from $s_i$ by the input sequences in the state identification set for $s_i$.

\begin{definition}\label{def:ball}
Given state $s_i$ of $M$ and FSM $M_I$ with state set $Q$,
$B(s_i) = \{q \in Q| \forall w \in \wset{s_i}. s_i \notseparates_w q \}$.
\end{definition}

\begin{example}\label{ex:ball}
Consider the $\mathit{Specification}$ and $\mathit{Implementation}_0$  FSMs in Example \ref{ex:thermostat-discretised} (Figure \ref{fig:thermostat-spec-discretised}), comprising the set of states $\{\mathit{Off}, \mathit{On}, \mathit{Off}_0, \mathit{On}_0\}$.
We therefore have that $Q = \{\mathit{On}_0, \mathit{Off}_0\}$.
Given the identification set $\{s\}$, we have that
$B(\mathit{Off}) =  \{\mathit{Off}_0\}$
and $B(\mathit{On}) = \{\mathit{On}_0\}$.
Clearly we have that $B(\mathit{On})$ and $B(\mathit{Off})$ are disjoint.
\end{example}

Observe that $B(s_i)$ is defined in terms of separability and \emph{not} strong separability. 
However, the notion of state identifier used in Definition~\ref{def:ball}
is defined in terms of strong separability (see Definition~\ref{def:strong_sep}).
The following shows what can happen if we were to instead use input sequences that separate, but do not strongly separate, the states of the specification.

\begin{example}\label{ex:ball_not}
Consider again the $\mathit{Specification}$ and $\mathit{Implementation}_0$  FSMs in Example \ref{ex:thermostat-discretised} (Figure \ref{fig:thermostat-spec-discretised}), comprising the set of states $\{\mathit{Off}, \mathit{On}, \mathit{Off}_0, \mathit{On}_0\}$,
with $Q = \{\mathit{On}_0, \mathit{Off}_0\}$.
Now consider the input $r$, which separates the states of the specification, and the sets
$B_1 = \{q \in Q| q \sim_{r} \mathit{Off}\}$
and
$B_2 = \{q \in Q| q \sim_{r} \mathit{On}\}$.
We have that
$B_1 = \{ \mathit{Off}_0 \}$
and
$B_2 = \{ \mathit{On}_0, \mathit{Off}_0 \}$.
In this case, $B_1$ and $B_2$ are not disjoint.
\end{example}

The following is an important consequence of strong separability and tells us that if we take two strongly separable states $s_1$ and $s_2$ of $M$ then $B(s_1)$ and $B(s_2)$ are pairwise disjoint (ie the property we observed in Example~\ref{ex:ball} always holds).

\begin{proposition}\label{prop_apart_implies_sep}
Given states $s_1$ and $s_2$ of $M$ with $s_1 \wapart_{\bar{x}} s_2$,
if $s'_1 \in B_{\bar{x}}(s_1)$ and $s'_2 \in B_{\bar{x}}(s_2)$
then $s'_1 \neq s'_2$.
\end{proposition}

\begin{proof}
We use a proof by contradiction, assuming that there is a state $q$ of $M_I$ such that $q \in B_{\bar{x}}(s_1) \cap B_{\bar{x}}(s_2)$.
By the definition of strong separability (Definition~\ref{def:strong_sep}), since
$s_1 \wapart_{\bar{x}} s_2$, either
$q \separates_{\bar{x}} s_1$
or
$q \separates_{\bar{x}} s_2$.
The first of these contradicts 
$q \in B_{\bar{x}}(s_1)$
and the second contradicts
$q \in B_{\bar{x}}(s_2)$.
The result thus follows.
\end{proof}

The test generation approach will be based on placing a lower bound on the number of  states $M_I$ must have if the SUT is faulty and it passes all the test cases in a given test suite. 
The approach will use
Proposition~\ref{prop_apart_implies_sep},
which tells us that if two input sequences $\bar{x}_1$ and $\bar{x}_2$  reach states $s_1$ and $s_2$ of $M$ respectively, $\bar{x}$ strongly separates $s_1$ and $s_2$ (ie $s_1 \wapart_{\bar{x}} s_2$), and the SUT conforms to the SUT on both $\bar{x}_1.\bar{x}$ and $\bar{x}_2.\bar{x}$ then $\bar{x}_1$ and $\bar{x}_2$ must reach \emph{different} states of $M_I$.

We now consider what we know about the states of the SUT reached by the state cover $V$ if $M_I$ conforms to $M$ on the set of input sequences formed by following each sequence in $V$ by the corresponding state identification set.

\begin{proposition}\label{prop_V_reach}
Let $V$ be a state cover for FSM $M$ whose states are pairwise strongly separable.
Further, let us suppose that every pair $(s_i,q_i)$ of states of the product machine $P(M,M_I)$
reached by a trace in $V$ satisfies the property that $s_i \notseparates_{\wset{s_i}} q_i$.
Then $V$ reaches $n$ separate
states of $M_I$.
\end{proposition}

\begin{proof}
Since $V$ contains $n$ traces,
it is sufficient to prove that if $v_1$ and $v_2$ are different input sequences in $V$ then they reach different states of $M_I$.
Let us suppose that $v_1 \neq v_2$ and also that $v_1$ and $v_2$ reach the states $(s_1,q_1)$ and $(s_2,q_2)$ of $P(M,M_I)$ respectively.
There is some witness $\witness{s_1}{s_2}$ that $s_1$ and $s_2$ are strongly separable and we will use $\bar{x}$ to denote $\witness{s_1}{s_2}$.
We know that $s_1 \wapart_{\bar{x}} s_2$,
$s_1 \notseparates_{\bar{x}} q_1$ and $s_2 \notseparates_{\bar{x}} q_2$.
But,
by Proposition~\ref{prop_apart_implies_sep},
this means that
$q_1 \neq q_2$
and so $v_1$ and $v_2$ reach different states of $M_I$.
\end{proof}

\begin{example}\label{ex:different-states}
Following up on Examples \ref{ex:state-cover}, \ref{ex:state-identification}, and \ref{ex:ball}, consider the state cover set $\{\varepsilon, s\}$ and the state identification set $\{s\}$ for which the states of specification are strongly separable. 
Now consider $\mathit{Implementation}_0$; the 
state cover set of the specification will take  $\mathit{Implementation}_0$ to $\mathit{Off}_0$ and 
$\mathit{On}_0$, which are clearly different states.

The same observation can be made for    $\mathit{Implementation}_1$. Note that by bringing  $\mathit{Specification}$
    and $\mathit{Implementation}_1$ using $s$ to $(\mathit{On}, \mathit{On}_1)$ and performing the state identification sequence $s$ from there, we arrive at $(\mathit{Off}, \mathit{On}_1)$, which are not similar and will be distinguished using one further input, i.e., by applying the state identification sequence. This is formalised in the subsequent  results. 
\end{example}

The test generation techniques that we build upon use input sequences of the form $v.\bar{x}.w$ such that $v$ is a sequence in the state cover $V$, $\bar{x}$ is any input sequence of a given length $\ell$ (determined by the maximum number of extra states in the SUT), and $w$ is a state identifier for the state $\delta(s_0,v.\bar{x})$ of $M$.
There are two ways in which such an input sequence can lead to a failure being observed.
\begin{enumerate}
\item The SUT $M_I$ does not conform to the specification $M$ on the initial sequence $v.\bar{x}$. If the SUT has output faults (one or more transitions produce the wrong output) then these output faults will lead to such failures.

\item
There are $v$, $\bar{x}$, and $w$ such that the SUT $M_I$ does conform to the specification $M$ on the initial sequence $v.\bar{x}$ but does not conform to the specification on $v.\bar{x}.w$.
In this case, a state identifier for $\delta(s_0,v.\bar{x})$ separates the state of the SUT $M_I$ reached by $v.\bar{x}$ from the state of the specification $M$ reached by $v.\bar{x}$.
\end{enumerate}

The following introduces corresponding notation, for a fixed length $\ell$ of input sequence $\bar{x}$.

\begin{definition}\label{def:d-dw}
Let us suppose that all states of FSM $M$ are strongly separable and $M_I$ is an FSM with the same input and output alphabets as $M$.
Further, let us suppose that $V$ is a state-cover for $M$ and for all $s_i \in S$ we have that $\wset{s_i}$ is a state identification set for $s_i$.
Given integer $\ell>0$:
\[
\begin{array}{lcl}
\diff{\ell} &=& 
\{(v,\bar{x}) | v \in V \wedge |\bar{x}| = \ell \wedge \\ && ~~M \separates_{v \bar{x}} M_I\} \\
\diffw{\ell} &=& 
\{(v,\bar{x}) | v \in V \wedge |\bar{x}| = \ell \wedge
\\
& & ~~\exists w \in \wset{\delta(s_0,v\bar{x})} . \\
&& ~~M \separates_{v \bar{x}w} M_I\}
\end{array}
\]
\end{definition}

We now provide two results that show what one can deduce from the value of the smallest integer $\ell$ such that 
$\diff{\ell} \cup \diffw{\ell} \neq \emptyset$.
Given this $\ell$,
we will
place a lower bound on the number of different states of $M_I$ reached by $V$ and sequences of the form $v.\bar{x}$, for $v \in V$ and input sequence $\bar{x}$ of length at most $\ell$
(Proposition~\ref{prop_min_length}).
We then show that
this implies an upper bound on the value of $\ell$ that we need to use in testing
(Proposition~\ref{corollary_find_failure}).

\begin{lemma}\label{lemma1}
Let us suppose that all states of FSM $M$ are strongly separable
and for all $v \in V$ and $w \in \wset{\delta(s_0,v)}$ we have that
$M \notseparates_{vw} M_I$.
Let us also suppose that
$M_I$ does not conform to $M$
and
$\ell >0$ is the smallest value such that
$\diff{\ell} \cup \diffw{\ell} \neq \emptyset$.
If
$(v,\bar{x}) \in \diff{\ell} \cup \diffw{\ell}$
and $\bar{x}_1$ is a non-empty proper prefix of $\bar{x}$
then the state of $M_I$
reached by $v\bar{x}_1$
is not a state of $M_I$ reached by
an input sequence in $V$.
\end{lemma}

\begin{proof}
Since $\bar{x}_1$ is a non-empty proper prefix of $\bar{x}$, there is some non-empty input sequence $\bar{x}_2$ such
that $\bar{x} = \bar{x}_1 \bar{x}_2$.
Further,
by the minimality of $\ell$,
$M_I$ must conform to $M$ on $v\bar{x}_1$ and so,
by Proposition~\ref{prop_error_state}, 
$v\bar{x}_1$ does not take the product machine to the error state.
Let us suppose that $(s_1,q_1)$ is the state of the product machine $P(M,M_I)$ reached by $v\bar{x}_1$.
By the minimality of $\ell$,
we must have that $q_1 \in B(s_1)$.

Let $v_2$ be some input sequence in $V$
and let us suppose that $v_2$ reaches the state $(s_2,q_2)$ of $PM(M,M_I)$.
It is therefore sufficient to prove that $q_1 \neq q_2$.
Since 
for all $v \in V$ and $w \in \wset{\delta(s_0,v)}$ we have that
$M \notseparates_{vw} M_I$,
we know that $q_2 \in B(s_2)$.

There are now two cases to consider.
\begin{enumerate}
\item Case 1: $s_1 \neq s_2$.
Let $w = \witness{s_1}{s_2}$.
We have that $s_1 \wapart_w s_2$,
$q_1 \in B(s_1)$ and $q_2 \in B(s_2)$.
By Proposition~\ref{prop_apart_implies_sep} we have that
$q_1 \neq q_2$ as required.
\item Case 2: $s_1 = s_2$.
We use proof by contradiction and assume that $q_1 = q_2$.
This means that the states of the product machine $P(M,M_I)$ reached by $v_2$ and $v\bar{x}_1$ are the same.

There are two cases.

\begin{enumerate}
\item Case 2a:
$(v,\bar{x}) \in \diff{\ell}$ and so
$M \separates_{v\bar{x}} M_I$.
By the minimality of $\ell$,
we have that 
$M \notseparates_{v\bar{x}_1} M_I$.
Since 
$M \separates_{v\bar{x}} M_I$
and $\bar{x} = \bar{x}_1 \bar{x}_2$
we therefore have that
$s_1 \separates_{\bar{x}_2} q_1$.
But,
since $v_2$ takes the product
machine $PM(M,M_I)$ to state
$(s_2,q_2) = (s_1,q_1)$,
we have that
$M \separates_{v_2\bar{x}_2} M_I$.
Since $\bar{x}_2$ is
shorter than $\bar{x}$,
this tells us that 
there is some $\ell'<\ell$
such that
$\diff{\ell} \neq \emptyset$.
This contradicts the minimality of $\ell$
as required.

\item Case 2b:
$(v,\bar{x}) \in \diffw{\ell}$ and so there exists $w \in \wset{s_1}$ such that
$M \separates_{v\bar{x}w} M_I$.
By the minimality of $\ell$,
we have that 
$M \notseparates_{v\bar{x}_1} M_I$.
Since 
$M \separates_{v\bar{x}w} M_I$
we therefore have that
$s_1 \separates_{\bar{x}_2w} q_1$.
But,
since $v_2$ takes the product
machine $PM(M,M_I)$ to state
$(s_2,q_2) = (s_1,q_1)$,
we have that
$M \separates_{v_2\bar{x}_2w} M_I$.
Since $\bar{x}_2$ is
shorter than $\bar{x}$,
there is some $\ell'<\ell$
such that
$\diffw{\ell'} \neq \emptyset$.
This contradicts the minimality of $\ell$
as required.
\end{enumerate}
\end{enumerate}
\end{proof}

\begin{lemma}\label{lemma2}
Let us suppose that all pairs of states of FSM $M$ are strongly separable
and for all $v \in V$ and $w \in \wset{\delta(s_0,v)}$ we have that
$M \notseparates_{vw} M_I$.
Further, let us suppose that
$M_I$ does not conform to $M$
and
$\ell >0$ is the smallest value such that
$\diff{\ell} \cup \diffw{\ell} \neq \emptyset$.
If
$(v,\bar{x}) \in \diff{\ell} \cup \diffw{\ell}$
and $\bar{x}_1, \bar{x}_2$ are different non-empty proper prefixes of $\bar{x}$
then
$v\bar{x}_1$ and $v\bar{x}_2$
reach different states of $M_I$.
\end{lemma}

\begin{proof}
First,
by Proposition~\ref{prop_error_state}, 
$v\bar{x}_1$ and $v\bar{x}_2$ do not take the product machine to the error state.
Let us suppose that $(s_1,q_1)$ is the state of the product machine $P(M,M_I)$ reached by $v\bar{x}_1$ and
$(s_2,q_2)$ is the state of the product machine $P(M,M_I)$ reached by $v\bar{x}_2$.
There are two cases to consider.

\begin{enumerate}
\item Case 1: $s_1 \neq s_2$.
Let $w = \witness{s_1}{s_2}$.
We have that $s_1 \wapart_w s_2$,
$q_1 \in B(s_1)$ and $q_2 \in B(s_2)$.
By Proposition~\ref{prop_apart_implies_sep} we have that
$q_1 \neq q_2$ as required.
\item Case 2: $s_1 = s_2$.
We use proof by contradiction and assume that $q_1 = q_2$.
We therefore have that
$v\bar{x}_1$ and $v\bar{x}_2$ reach the same states of the product machine $P(M,M_I)$.

Without loss of generality assume that $\bar{x}_1$ is a proper prefix of $\bar{x}_2$
and define $\bar{x}_3$ by $\bar{x} = \bar{x}_2 \bar{x}_3$.
There are two cases.

\begin{enumerate}
\item Case 2a:
$(v,\bar{x}) \in \diff{\ell}$ and so
$M \separates_{v\bar{x}} M_I$.
By the minimality of $\ell$,
we have that 
$M \notseparates_{v\bar{x}_2} M_I$.
Since 
$M \separates_{v\bar{x}} M_I$
and $\bar{x} = \bar{x}_2 \bar{x}_3$
we therefore have that
$s_2 \separates_{\bar{x}_3} q_2$.
But,
since $v\bar{x}_1$ takes the product
machine $PM(M,M_I)$ to state
$(s_1,q_1) = (s_2,q_2)$,
we have that
$M \separates_{v\bar{x}_1\bar{x}_3} M_I$.
Since $\bar{x}_1\bar{x}_3$ is
shorter than $\bar{x}$
this tells us that 
there is some $\ell'<\ell$
such that
$\diff{\ell'} \neq \emptyset$.
This contradicts the minimality of $\ell$
as required.

\item Case 2b:
$(v,\bar{x}) \in \diffw{\ell}$ and so there exists $w \in \wset{\delta(s_0,v\bar{x})}$ such that
$M \separates_{v\bar{x}w} M_I$.
By the minimality of $\ell$,
we have that 
$M \notseparates_{v\bar{x}_2} M_I$.
Since 
$M \separates_{v\bar{x}w} M_I$
we therefore have that
$s_2 \separates_{\bar{x}_3w} q_2$.
But,
since $v\bar{x}_1$ takes the product
machine $PM(M,M_I)$ to state
$(s_1,q_1) = (s_2,q_2)$,
we have that
$M \separates_{v\bar{x}_1\bar{x}_3w} M_I$.
Since $\bar{x}_1\bar{x}_3$ is
shorter than $\bar{x}$,
this tells us that 
there is some $\ell'<\ell$
such that
$\diffw{\ell'} \neq \emptyset$.
This contradicts the minimality of $\ell$
as required.
\end{enumerate}
\end{enumerate}
\end{proof}

\begin{example}
For the non-conforming implementation $\mathit{Implementation}_1$, we have that it is sufficient to take $\bar{x}$ to be $s$ since after $s \in V$ and $s \in W$, a further input $s$ can reach a non-conformance verdict between the specification and the implementation. This is because the faulty implementation does not have any additional state with respect to the specification. For faulty implementations with additional states, we may need longer intermediate  inputs to reach non-conformance but there is a fixed upper bound given an upper bound on the size of the implementation, as stated below.  
In terms of Definition \ref{def:d-dw} this is summarised as follows, where, and for brevity we refer to the specification as $M$ and to $\mathit{Implementation}_1$ as $M_I$:
\[
\begin{array}{lcl}
\diff{1} &=& \emptyset \\
\diffw{1} &=& \{  (s, s)\}
\end{array}
\]

Note that $\diff{1}$ is the empty set because the pair with the shortest sequence to distinguish $M$ and $M_I$ is $(s, s.s)$, of which the second component is beyond the allowed limit of $l = 1$. However, $\diffw{1}$ contains the pair $(s, s)$, because after $s.s$ by performing a further $s \in W$, we can distinguish the two systems. 
\end{example}

We can use the above-given Lemmas to place a lower bound on the number of states an SUT must have if the SUT is faulty and we know the smallest value $\ell$ such that $\diff{\ell} \cup \diffw{\ell} \neq \emptyset$.

\begin{proposition}\label{prop_min_length}
Let us suppose that all pairs of states of FSM $M$ are strongly separable, 
$M_I$ does not conform to $M$,
and for all $v \in V$ and $w \in \wset{\delta(s_0,v)}$ we have that
$M \notseparates_{vw} M_I$.
If $\ell$ is the smallest value such that 
$\diff{\ell} \cup \diffw{\ell} \neq \emptyset$
then $M_I$ has at least $n+\ell-1$ states.
\end{proposition}

\begin{proof}
First recall that,
by Proposition~\ref{prop_V_reach},
$V$ reaches
$n$ separate states of $M_I$.
Let $(v,\bar{x})$ be some pair in 
$\diff{\ell} \cup \diffw{\ell}$.

By Lemma~\ref{lemma1},
if $\bar{x}_1$ is a non-empty proper prefix of $\bar{x}$ then the state of $M_I$
reached by $v\bar{x}_1$ is not a state of $M_I$ reached by $V$.
Thus the number of states of $M_I$
must be at least $n = |V|$
plus the number of states of $M_I$
reached by sequences of the form
$v\bar{x}_1$ where $\bar{x}_1$
is a non-empty proper prefix of $\bar{x}$.

We now need to consider how many different states of $M_I$ are reached by 
sequences of the form
$v\bar{x}_1$ where $\bar{x}_1$
is a non-empty proper prefix of $\bar{x}$.
Lemma~\ref{lemma2} tells us that if
$\bar{x}_1$ and $\bar{x}_2$ are different non-empty proper prefixes of $\bar{x}$
then $v\bar{x}_1$ and $v\bar{x}_2$
reach different states of $M_I$.
Thus, 
the number of
different states of $M_I$ reached by 
sequences of the form
$v\bar{x}_1$,
where $\bar{x}_1$
is a non-empty proper prefix of $\bar{x}$,
is equal to the number of different non-empty proper prefixes of $\bar{x}$ and that is $|\bar{x}|-1$.

Bringing these together,
we have that the sequences in $V$ reach $n$ different states of $M_I$ and the sequences of the form $v \bar{x}_1$,
where $\bar{x}_1$ is a non-empty proper prefix of $\bar{x}$,
reach another $|\bar{x}| - 1$ states of $M_I$.
Thus, $M_I$ must have at least $n+|\bar{x}|-1$ states.
It is now sufficient to observe that, by definition, $|\bar{x}| = \ell$.   
\end{proof}

Now consider the standard FSM-based testing context in which
we assume that the FSM $M_I$ that models the SUT has at most $n+k$ states
and we wish to generate sufficient tests to determine whether $M_I$ conforms to $M$.
We can use the above result
(Proposition~\ref{prop_min_length})
to provide an upper bound on the value of $\ell$ that we require to use when testing with sequences of the form
$v.\bar{x}$ and $v.\bar{x}.w$, with $|\bar{x}| \leq \ell$.

\begin{proposition}\label{corollary_find_failure}
Let us suppose that all pairs of states of FSM $M$ are strongly separable, $M_I$ does not conform to $M$, and for all $v \in V$ and $w \in \wset{\delta(s_0,v)}$ we have that
$M \notseparates_{vw} M_I$.
If $M_I$ has at most $k$ more states than $M$ then there exists some $v \in V$,
and $\bar{x} \in X^*$
such that $|\bar{x}| \leq k+1$ and
one of the following holds:
\begin{itemize}
\item $M_I \separates_{v\bar{x}} M$; or
\item there is some $w \in \wset{\delta(s_0,v\bar{x})}$
such that
$M_I \separates_{v\bar{x}w} M$.
\end{itemize}
\end{proposition}

\begin{proof}
This follows from Proposition~\ref{prop_min_length} through a proof by contradiction.
\end{proof}

The following shows what this implies regarding test completeness.

\begin{theorem}
Let us suppose that all pairs of states of FSM $M$ are strongly separable, $M$ has $n$ states, and $M_I$ has at most $m=n+k$ states.
If
$\diff{\ell} = \emptyset$ and
$\diffw{\ell} = \emptyset$ 
for all $\ell \leq k+1$
then $M_I$ conforms to $M$. 
\end{theorem}

\begin{proof}
This follows from 
Proposition~\ref{corollary_find_failure}.
\end{proof}

This result shows how test generation can proceed:
one simply uses the test sequences implicit in the definitions of 
$\diff{\ell}$ and $\diffw{\ell}$. 
Algorithm~\ref{alg:HSI} gives the corresponding test generation algorithm.
Note that for a state $s_i$ of $M$, $\wset{s_i}$ is a Harmonized State Identifier for $s_i$ as used in the HSI-method except that we require $\wset{s_i}$ to strongly separate $s_i$ from other states of $M$; it is not sufficient that
$\wset{s_i}$ to separates $s_i$ from other states of $M$.
As a result, 
Algorithm~\ref{alg:HSI} is essentially the HSI-method, with the only difference being how we define conformance and the  $\wset{s_i}$.
Since the test set produced by the HSI-method is a subset of those produced by the W-method and the Wp-method, the above demonstrates that these test generation techniques also produce $m$-complete test sets when we have a notion of similarity of outputs rather than equality and we use state identifiers that strongly separate states.

\begin{algorithm}[t]
\caption{Test Generation: The HSI Method Using Strong Separability}
\label{alg:HSI}
\begin{algorithmic}[0]
\State{Input: FSM $M=(S,s_0,X,Y,\delta,\lambda)$ where $S = \{s_1, \ldots, s_n\}$, $m \geq n$}
\State{Derive state cover $V$}
\State{Derive state identification sets $\wset{s_1}, \ldots, \wset{s_n}$ based on strong separability}
\State{Derive $T = \{v.\bar{x}.w | v \in V, \bar{x} \in X^*, 0 \leq |\bar{x}| \leq m-n+1, w \in \wset{\delta(s_0,v.\bar{x})}\}$}
\State{Remove from $T$ all prefixes (ie $T = RP(T)$)}
\State{\textbf{return} $T$}
\end{algorithmic}
\end{algorithm}

\begin{example}
Consider the $\mathit{Specification}$ FSM in Figure \ref{fig:thermostat-spec-discretised}. We have already established in Example \ref{ex:state-cover} that $V = \{\varepsilon, s\}$ and in Example \ref{ex:state-identification} that the state-identification set that strongly separates the specification states is $W = \{s\}$. 
Hence, assuming that we are only interested in implementations that have the same state count as the specification, by applying Algorithm \ref{alg:HSI}, we obtain: 
\[
\begin{array}{lcl}
T &=& RP(\{v.\bar{x}.w | v \in \{\varepsilon, s\} \wedge \bar{x} \in \{\varepsilon, r, s\} \wedge w \in \{s\}\}) \\
T &=& RP(\{\varepsilon.\varepsilon.s, \varepsilon.r.s, \varepsilon.s.s, s.r.s, s.s.s\}) \\
&=& RP(\{s, r.s, s.s, s.r.s, s.s.s\}) \\
&=& \{r.s, s.r.s, s.s.s\}
\end{array}
\]
Note that the key difference between this set of generated inputs and the traditional HSI method is that the adapted algorithm forces the state-identification set to provide strong separability. 
Otherwise, the traditional HSI method can use $r$ to separate states and generate $r.r$, $s.s.r$ and $s.r.r$ as the longest sequences and these sequences will miss the fault in $\mathit{Implementation}_1$. 
This demonstrates that the HSI method, with state identifiers that do not strongly separate the states of the specification, is not complete.
Using this specification FSM, state cover, and state identifiers, the W-method and Wp-method return the same test suite and so are also incomplete. 
\end{example}

It is now worth considering how large the test suite can be for an FSM specification $M$ with $n$ states.
Clearly $V$ contains $n$ input sequences and these have length at most $n-1$: we can generate such input sequences through a breadth-first search.
Further, each $\wset{s_i}$ contains at most $n-1$ sequences.
The following places an upper bound on the lengths of sequences in $\wset{s_i}$, assuming we use shortest sequences that suffice.

\begin{proposition}\label{prop:bound_W_length}
Let us suppose that $M$ is an FSM with $n$ states and all pairs of states of $M$ are strongly separable.
For every pair $s_1, s_2$ of distinct states of $M$, there is an input sequence of length at most $\frac{n(n-1)}{2}$ that strongly separates $s_1$ and $s_2$.
\end{proposition}

\begin{proof}
Given any pair $(s_1,s_2)$ of distinct states of $M$ let $d(s_1,s_2)$ denote the length of the shortest input sequence(s) that strongly separate $s_i$ and $s_j$.

Choose two states $s_1$ and $s_2$ and let us suppose that $d(s_1,s_2) = a$. It is sufficient to prove that $a \leq \frac{n(n-1)}{2}$.

Let $\bar{x} = x_1 \ldots x_a$ be a shortest input sequence that strongly separates $s_1$ and $s_2$.
For $i \in \{1,2\}$ and $1 \leq j < a$ define $s_i^j = \delta(s_i, x_1 \ldots x_j)$.
Consider some $1 \leq j < a$ and states $s_1^j$ and $s_2^j$.
Since 
$\bar{x} = x_1 \ldots x_a$ is a shortest input sequence that strongly separates $s_1$ and $s_2$,
we must have that 
$x_{j+1} \ldots x_a$ is a shortest input sequence that strongly separates $s_1^j$ and $s_2^j$.
By the minimality of these input sequences, for all $1 \leq j < k < a$ we have that
$\{s_1^i,s_2^i\} \neq \{s_1^j,s_2^j\}$.
Thus, the length of $\bar{x}$ is bounded by the number of unordered pairs of distinct states of $M$, which is $\frac{n(n-1)}{2}$.
The result therefore follows.
\end{proof}

Thus, the sizes of the state cover and state identification sets are low-order polynomial.
Naturally, the size of the test suite grows exponentially as the upper bound, $k$, on the number of extra states grows.
However, this is also the case with techniques such as the HSI-method when applied in the classical FSM context.
In practice, the value of $k$ used might depend upon domain knowledge and a cost/benefit analysis.

It is finally worth commenting on the computational complexity of deriving the state cover and state identifiers. Clearly, a state cover can be derived in low-order polynomial time since breadth-first search takes linear time.
It is also possible
to derive the state identifiers in polynomial time.
In particular, it is possible to define a simple iterative algorithm based on the proof of Proposition~\ref{prop:bound_W_length},
with this operating as follows.
First, we determine which pairs of states can be strongly separated by a single input and record these (we could say that these are strongly $1$-separable).
We then determine which of the remaining pairs $\{s_i,s_j\}$ of states are strongly separated by input sequences of length $2$ (they are strongly $2$-separable).
In order to check this, for each pair $\{s_i,s_j\}$ of states (that are not strongly $1$-separable) we check whether there is an input $x$ such that states $\delta(s_i,x)$ and $\delta(s_j,x)$ are strongly $1$-separable.
This process continues to iterate:
in each iteration,
for each pair $\{s_i,s_j\}$ of states not already strongly separated, we determine whether there is an input $x$ such that states $\delta(s_i,x)$ and $\delta(s_j,x)$ have already been strongly separated in an earlier iteration.
Clearly, each iteration takes polynomial time and, from Proposition~\ref{prop:bound_W_length}, we know that there are at most  $\frac{n(n-1)}{2}$ iterations.


\section{Conclusions}\label{section:conclusions}
In this paper, we proposed the notion of strong separability, which is inspired by apartness in constructive mathematics. This notion allows for separating states according to approximate notions of conformance, i.e., notions that allow for a margin of error, in quantitative finite-state machine models of systems. We showed the applicability of our notion by adopting it in a well-known model-based testing technique, called the Harmonized State Identifier (HSI) method. We proved that using our notion, the HSI method is complete for approximate notions of conformance. We exemplified and illustrated our approach through a simple example of a thermostat and its correct and incorrect implementations. 
We also used this example to demonstrate that complete test suites need not be produced if we use separability rather than strong separability.

Recall that the approach described in this paper operates by reasoning about the number of different states of the SUT met during testing and does this by strongly separating states.
We have seen that this is consistent with the W-method, the Wp-method and the HSI-method.
These three approaches differ in the choice of state identification sequences used but they keep this choice fixed throughout test generation.
Dorofeeva et al. \cite{DorofeevaEY05} introduced the  H-method,
which allows
different state identification sequences to be used for different input sequences that follow sequences from $V$.
It may be possible to extend the results given in this paper to the H-method by, instead of using a fixed input sequence
$\witness{s_1}{s_2}$ to separate two states of $M$, allowing the choice of input sequence to vary.
Specifically, if we are considering an input sequence $v\bar{x}$ where $v \in V$
then the $\witness{s_1}{s_2}$ used could depend on $v$ and $\bar{x}$.

There are a number of other lines of possible future work that correspond to the following limitations of the approach described in this paper.
First, we abandoned the use of the observation tree 
used by Vaandrager~\cite{Vaandrager24}
since a correct SUT might produce traces that are not traces of the specification $M$;
it may well be that it is possible to reason about observation trees that are `similar' to an observation tree of $M$.
Second,
this paper has concerned deterministic, completely-specified FSMs and the test generation algorithm assumes that the states of $M$ are pairwise strongly separable. Although these requirements are weaker than those imposed by the work that uses apartness, there is the interesting question of how they might be further weakened. In particular, there is the challenge of developing test generation algorithms that return $m$-complete test suites when some states of $M$ are not strongly separable and also when $M$ is partial and/or non-deterministic.

Our adaptation of apartness in the quantitative setting, enables the application of recent automata learning algorithms to this setting as well. 
We will be looking into using our 
stronger notion of apartness to develop novel automata learning algorithms by integrating algorithms that use apartness \cite{VaandragerGRW22,VaandragerS25} with those that consider quantitative extensions of state machines \cite{MedhatRBF15,BacciILR21,TapplerA0EL21}. 

\section*{Acknowledgments}

Robert M. Hierons was
partially supported by the UKRI Trustworthy Autonomous Systems Node in Verifiability, Grant Award Reference EP/V026801/2 and the
EPSRC grant RoboTest: Systematic Model-Based Testing and Simulation of Mobile Autonomous Robots, Grant Award Reference EP/R025134/1.
Mohammad Reza Mousavi have been partially supported by the UKRI Trustworthy Autonomous Systems Node in Verifiability, Grant Award Reference EP/V026801/2 and the EPSRC grant on  Verified Simulation for Large Quantum Systems (VSL-Q), Grant Award Reference EP/Y005244/1.

\bibliographystyle{plain}
\bibliography{refs}

\begin{thebibliography}{10}

\bibitem{AbbasMF14}
Houssam Abbas, Hans~D. Mittelmann, and Georgios Fainekos.
\newblock Formal property verification in a conformance testing framework.
\newblock In {\em Twelfth {ACM/IEEE} International Conference on Formal Methods
  and Models for Codesign, {MEMOCODE} 2014, Lausanne, Switzerland, October
  19-21, 2014}, pages 155--164. {IEEE}, 2014.

\bibitem{BacciILR21}
Giovanni Bacci, Anna Ing{\'{o}}lfsd{\'{o}}ttir, Kim~G. Larsen, and
  Rapha{\"{e}}l Reynouard.
\newblock Active learning of markov decision processes using baum-welch
  algorithm.
\newblock In M.~Arif Wani, Ishwar~K. Sethi, Weisong Shi, Guangzhi Qu,
  Daniela~Stan Raicu, and Ruoming Jin, editors, {\em 20th {IEEE} International
  Conference on Machine Learning and Applications, {ICMLA} 2021, Pasadena, CA,
  USA, December 13-16, 2021}, pages 1203--1208. {IEEE}, 2021.

\bibitem{BroyEtAl2005}
Manfred Broy, Bengt Jonsson, Joost-Pieter Katoen, Martin Leucker, and Alexander
  Pretschner.
\newblock {\em Model-Based Testing of Reactive Systems}, volume 3472 of {\em
  Lecture Notes in Computer Science}.
\newblock Springer, 2005.

\bibitem{chow78}
Tsun~S. Chow.
\newblock Testing software design modelled by finite state machines.
\newblock {\em IEEE Transactions on Software Engineering}, 4:178--187, 1978.

\bibitem{AraujoH0V20}
Hugo~Leonardo da~Silva~Araujo, Ties Hoenselaar, Mohammad~Reza Mousavi, and
  Alexey~V. Vinel.
\newblock Connected automated driving: {A} model-based approach to the analysis
  of basic awareness services.
\newblock In {\em 31st {IEEE} Annual International Symposium on Personal,
  Indoor and Mobile Radio Communications, {PIMRC} 2020, London, United Kingdom,
  August 31 - September 3, 2020}, pages 1--7. {IEEE}, 2020.

\bibitem{AlfaroFS09}
Luca de~Alfaro, Marco Faella, and Mari{\"{e}}lle Stoelinga.
\newblock Linear and branching system metrics.
\newblock {\em {IEEE} Trans. Software Eng.}, 35(2):258--273, 2009.

\bibitem{DESHARNAIS2002163}
Jos{\'e}e Desharnais, Abbas Edalat, and Prakash Panangaden.
\newblock Bisimulation for labelled markov processes.
\newblock {\em Information and Computation}, 179(2):163--193, 2002.

\bibitem{DeshmukhMP17}
Jyotirmoy~V. Deshmukh, Rupak Majumdar, and Vinayak~S. Prabhu.
\newblock Quantifying conformance using the skorokhod metric.
\newblock {\em Formal Methods Syst. Des.}, 50(2-3):168--206, 2017.

\bibitem{DorofeevaEY05}
Rita Dorofeeva, Khaled El{-}Fakih, and Nina Yevtushenko.
\newblock An improved conformance testing method.
\newblock In {\em 25th {IFIP} {WG} 6.1 International Conference on Formal
  Techniques for Networked and Distributed Systems ({FORTE} 2005)}, volume 3731
  of {\em Lecture Notes in Computer Science}, pages 204--218. Springer, 2005.

\bibitem{fujiwara91}
S.~Fujiwara, G.~v.~Bochmann, F.~Khendek, M.~Amalou, and A.~Ghedamsi.
\newblock Test selection based on finite state models.
\newblock {\em IEEE Transactions on Software Engineering}, 17(6):591--603,
  1991.

\bibitem{gaudel95}
Marie-Claude Gaudel.
\newblock Testing can be formal too.
\newblock In {\em 6th International Joint Conference CAAP/FASE Theory and
  Practice of Software Development ({TAPSOFT'95})}, volume 915 of {\em Lecture
  Notes in Computer Science}, pages 82--96. Springer, 1995.

\bibitem{hennie64}
F.~C. Hennie.
\newblock Fault-detecting experiments for sequential circuits.
\newblock In {\em Proceedings of Fifth Annual Symposium on Switching Circuit
  Theory and Logical Design}, pages 95--110, Princeton, New Jersey, November
  1964.

\bibitem{HieronsBBCDDGHKKLSVWZ09}
Robert~M. Hierons, Kirill Bogdanov, Jonathan~P. Bowen, Rance Cleaveland, John
  Derrick, Jeremy Dick, Marian Gheorghe, Mark Harman, Kalpesh Kapoor, Paul
  Krause, Gerald L{\"u}ttgen, Anthony J.~H. Simons, Sergiy~A. Vilkomir,
  Martin~R. Woodward, and Hussein Zedan.
\newblock Using formal specifications to support testing.
\newblock {\em ACM Computing Surveys}, 41(2):9:1--9:76, 2009.

\bibitem{HMM09}
Robert~M. Hierons, Mercedes~G. Merayo, and Manuel N{\'u}{\~n}ez.
\newblock Testing from a stochastic timed system with a fault model.
\newblock {\em Journal of Logic and Algebraic Programming}, 78(2):98--115,
  2009.

\bibitem{KhakpourM15}
Narges Khakpour and Mohammad~Reza Mousavi.
\newblock Notions of conformance testing for cyber-physical systems: Overview
  and roadmap (invited paper).
\newblock In Luca Aceto and David de~Frutos{-}Escrig, editors, {\em 26th
  International Conference on Concurrency Theory, {CONCUR} 2015, Madrid, Spain,
  September 1.4, 2015}, volume~42 of {\em LIPIcs}, pages 18--40. Schloss
  Dagstuhl - Leibniz-Zentrum f{\"{u}}r Informatik, 2015.

\bibitem{lee94}
David Lee and Mihalis Yannakakis.
\newblock Testing finite-state machines: State identification and verification.
\newblock {\em {IEEE} Transactions on Computers}, 43(3):306--320, 1994.

\bibitem{lee96}
David Lee and Mihalis Yannakakis.
\newblock Principles and methods of testing finite-state machines - a survey.
\newblock {\em Proceedings of the {IEEE}}, 84(8):1089--1123, 1996.

\bibitem{luo94a}
G.~Luo, A.~Petrenko, and G.~v.~Bochmann.
\newblock Selecting test sequences for partially-specified nondeterministic
  finite state machines.
\newblock In {\em The 7th {IFIP} Workshop on Protocol Test Systems}, pages
  95--110, Tokyo, Japan, November 8--10 1994. Chapman and Hall.

\bibitem{MedhatRBF15}
Ramy Medhat, S.~Ramesh, Borzoo Bonakdarpour, and Sebastian Fischmeister.
\newblock A framework for mining hybrid automata from input/output traces.
\newblock In Alain Girault and Nan Guan, editors, {\em 2015 International
  Conference on Embedded Software, {EMSOFT} 2015, Amsterdam, Netherlands,
  October 4-9, 2015}, pages 177--186. {IEEE}, 2015.

\bibitem{MillerS12}
Tim Miller and Paul~A. Strooper.
\newblock A case study in model-based testing of specifications and
  implementations.
\newblock {\em Software Testing, Verification and Reliability}, 22(1):33--63,
  2012.

\bibitem{Shafie2022}
Muhammad~Luqman Mohd-Shafie, Wan Mohd Nasir~Wan Kadir, Horst Lichter, Muhammad
  Khatibsyarbini, and Mohd~Adham Isa.
\newblock Model-based test case generation and prioritization: a systematic
  literature review.
\newblock {\em Software and Systems Modeling}, 21(2):717--753, 2022.

\bibitem{moore56}
E.~F. Moore.
\newblock Gedanken-experiments.
\newblock In C.~Shannon and J.~McCarthy, editors, {\em Automata Studies}.
  Princeton University Press, 1956.

\bibitem{QinEtAl2023}
Xin Qin, Navid Hashemi, Lars Lindemann, and Jyotirmoy~V. Deshmukh.
\newblock Conformance testing for stochastic cyber-physical systems.
\newblock In Alexander Nadel and Kristin~Yvonne Rozier, editors, {\em Formal
  Methods in Computer-Aided Design, {FMCAD} 2023, Ames, IA, USA, October 24-27,
  2023}, pages 294--305. {IEEE}, 2023.

\bibitem{TapplerA0EL21}
Martin Tappler, Bernhard~K. Aichernig, Giovanni Bacci, Maria Eichlseder, and
  Kim~G. Larsen.
\newblock L\({}^{\mbox{*}}\)-based learning of markov decision processes
  (extended version).
\newblock {\em Formal Aspects Comput.}, 33(4-5):575--615, 2021.

\bibitem{Tretmans08}
Jan Tretmans.
\newblock Model based testing with labelled transition systems.
\newblock In {\em Formal Methods and Testing}, volume 4949 of {\em Lecture
  Notes in Computer Science}, pages 1--38. Springer, 2008.

\bibitem{UttingPL12}
Mark Utting, Alexander Pretschner, and Bruno Legeard.
\newblock A taxonomy of model-based testing approaches.
\newblock {\em Software Testing, Verification and Reliability}, 22(5):297--312,
  2012.

\bibitem{Vaandrager24}
Frits~W. Vaandrager.
\newblock A new perspective on conformance testing based on apartness.
\newblock In Venanzio Capretta, Robbert Krebbers, and Freek Wiedijk, editors,
  {\em Logics and Type Systems in Theory and Practice - Essays Dedicated to
  Herman Geuvers on The Occasion of His 60th Birthday}, volume 14560 of {\em
  Lecture Notes in Computer Science}, pages 225--240. Springer, 2024.

\bibitem{VaandragerGRW22}
Frits~W. Vaandrager, Bharat Garhewal, Jurriaan Rot, and Thorsten Wi{\ss}mann.
\newblock A new approach for active automata learning based on apartness.
\newblock In Dana Fisman and Grigore Rosu, editors, {\em Tools and Algorithms
  for the Construction and Analysis of Systems - 28th International Conference,
  {TACAS} 2022, Held as Part of the European Joint Conferences on Theory and
  Practice of Software, {ETAPS} 2022, Munich, Germany, April 2-7, 2022,
  Proceedings, Part {I}}, volume 13243 of {\em Lecture Notes in Computer
  Science}, pages 223--243. Springer, 2022.

\bibitem{VaandragerS25}
Frits~W. Vaandrager and Martijn Sanders.
\newblock L\({}^{\mbox{{\#}}}\) for dfas.
\newblock In Nils Jansen, Sebastian Junges, Benjamin~Lucien Kaminski, Christoph
  Matheja, Thomas Noll, Tim Quatmann, Mari{\"{e}}lle Stoelinga, and Matthias
  Volk, editors, {\em Principles of Verification: Cycling the Probabilistic
  Landscape - Essays Dedicated to Joost-Pieter Katoen on the Occasion of His
  60th Birthday, Part {III}}, volume 15262 of {\em Lecture Notes in Computer
  Science}, pages 155--172. Springer, 2024.

\bibitem{BreugelHMW05}
Franck van Breugel, Claudio Hermida, Michael Makkai, and James Worrell.
\newblock An accessible approach to behavioural pseudometrics.
\newblock In Lu{\'{\i}}s Caires, Giuseppe~F. Italiano, Lu{\'{\i}}s Monteiro,
  Catuscia Palamidessi, and Moti Yung, editors, {\em Automata, Languages and
  Programming, 32nd International Colloquium, {ICALP} 2005, Lisbon, Portugal,
  July 11-15, 2005, Proceedings}, volume 3580 of {\em Lecture Notes in Computer
  Science}, pages 1018--1030. Springer, 2005.

\bibitem{BreugelW01}
Franck van Breugel and James Worrell.
\newblock Towards quantitative verification of probabilistic transition
  systems.
\newblock In Fernando Orejas, Paul~G. Spirakis, and Jan van Leeuwen, editors,
  {\em Automata, Languages and Programming, 28th International Colloquium,
  {ICALP} 2001, Crete, Greece, July 8-12, 2001, Proceedings}, volume 2076 of
  {\em Lecture Notes in Computer Science}, pages 421--432. Springer, 2001.

\bibitem{vasilevskii73}
M.~P. Vasilevskii.
\newblock Failure diagnosis of automata.
\newblock {\em Cybernetics}, 4:653--665, 1973.

\end{thebibliography}

\end{document}